\documentclass[journal]{IEEEtran}

\usepackage{amsmath,amsthm}
\usepackage{amssymb}
\usepackage{amscd}
\usepackage{textcomp}
\usepackage{hyperref}
\usepackage{fancyvrb}
\usepackage{wasysym}

\usepackage{graphicx}

\usepackage{listings}
\usepackage{cite}

\usepackage{color}
\definecolor{red}{rgb}{1,0,0}
\definecolor{blue}{rgb}{0,0,1}
\definecolor{green}{rgb}{0,0.5,0}
\definecolor{magenta}{rgb}{1,0,1}

\newsavebox{\ieeealgbox}

\newcommand{\be}{\begin{equation}}
\newcommand{\ee}{\end{equation}}
\newcommand{\bea}{\begin{eqnarray}}
\newcommand{\eea}{\end{eqnarray}}
\newcommand{\bal}{\begin{align}}
\newcommand{\eal}{\end{align}}
\newcommand{\nn}{\nonumber}
\newcommand{\eye}{\mbox{$\mbox{1}\!\mbox{l}\;$}}

\renewcommand{\vec}[1]{\boldsymbol{#1}}
\newcommand{\DD}{\mathcal{D}}

\newcommand{\U}{\mathcal{U}}

\newtheorem{prop}{Proposition}

\newtheorem{lemma}{Lemma}

\newtheorem{corr}{Corollary}

\begin{document}
\title{Dual theory of transmission line outages}

\author{Henrik~Ronellenfitsch, %
Debsankha Manik,%
\thanks{H.~Ronellenfitsch and D.~Manik are at the Max Planck Institute for Dynamics and Self-Organization (MPIDS), 37077 G\"ottingen, Germany.
H.~Ronellenfitsch is additionally at the Department of Physics and
Astronomy, University of Pennsylvania, Philadelphia PA, USA.}
Jonas H\"orsch, %
Tom Brown,%
\thanks{J.~H\"orsch and T.~Brown are at the Frankfurt Institute for Advanced Study, 60438 Frankfurt am Main, Germany.}
Dirk~Witthaut%
\thanks{D.~Witthaut is at the Forschungszentrum J\"ulich, Institute for Energy and Climate Research -
	Systems Analysis and Technology Evaluation (IEK-STE),  52428 J\"ulich, Germany
and the Institute for Theoretical Physics, University of Cologne,
50937 K\"oln, Germany.}}


\markboth{Journal of \LaTeX\ Class Files,~Vol.~13, No.~9, September~2014}%
{Shell \MakeLowercase{\textit{et al.}}: Bare Demo of IEEEtran.cls for Journals}

\maketitle

\begin{abstract}
  A new graph dual formalism is presented for the analysis of line outages
  in electricity networks. The dual formalism is based on a
  consideration of the flows around closed cycles in the network. After
  some exposition of the theory is presented, a
  new formula for the computation of Line Outage Distribution Factors
  (LODFs) is derived, which is not only computationally faster than
  existing methods, but also generalizes easily for
  multiple line outages and arbitrary changes to line series reactance. In addition, the dual formalism provides new physical
  insight for how the effects of line outages propagate through the
  network. For example, in a planar network a single line outage can be shown to induce
  monotonically decreasing flow changes, which are
  mathematically equivalent to an electrostatic dipole field.
\end{abstract}


\IEEEpeerreviewmaketitle


\section{Introduction}\label{sec:intro}

The robustness of the power system relies on its ability to withstand
disturbances, such as line and generator outages. The grid is usually
operated with `$n-1$ security', which means that it should withstand
the failure of any single component, such as a transmission circuit or
a transformer. The analysis of such contingencies has gained in
importance with the increasing use of generation from variable
renewables, which have led to larger power imbalances in the grid and
more situations in which transmission lines are loaded close to their
thermal limits \cite{Amin05,Heid10,12powergrid,Pesc14,Brown2016,16redundancy}.

A crucial tool for contingency analysis is the use of Line Outage Distribution Factors (LODFs), which measure the linear sensitivity of active power flows in the network to outages of specific lines \cite{Wood14}. LODFs are not only used to calculate power flows after an outage, but are also employed in security-constrained linear optimal power flow (SCLOPF), where power plant dispatch is optimized such that the network is always $n-1$ secure \cite{Wood14}.

LODF matrices can be calculated from Power Transfer Distribution
Factors (PTDFs) \cite{Gule07,Guo09}, which describe how power flows change when
power injection is shifted from one node to another. In
\cite{ronellen16}, a dual method for calculating PTDFs was
presented. The dual method is based on an analysis of the flows around
closed cycles (closed cycles are paths that are non-intersecting and start and
end at the same node~\cite{Dies10}) in the network graph; for a plane graph, a basis of
these closed cycles corresponds to the nodes of the dual graph~\cite{Dies10}.
Seen as a planar polygon, each basis cycle corresponds to one facet of the polygon.
(This notion of dual of a plane graph is called the \emph{weak dual}).
In this paper the dual formalism is applied to derive a
new direct formula for the LODF matrices, which is not only computationally faster than existing
methods, but has several other advantages. It can be easily extended
to take account of multiple line outages and, unlike other methods, it also works for the case
where a line series reactance is modified rather than failing completely.
This latter property is relevant given the increasing use of controllable series compensation devices
for steering network power flows. Moreover, the dual formalism
is not just a calculational tool: it provides new insight
into the physics of how the effects of outages propagate in the
network, which leads to several useful results.

Depending on network topology, the dual method can
lead to a significant improvement in the speed
of calculating LODFs. Thus it can be useful for
applications where LODFs must be calculated repeatedly and
in a time-critical fashion, for
instance in `hot-start DC models' or `incremental DC models'~\cite{Stot09}.
The dual method we describe is particularly suited to these types of
problems because unlike in the primal case, most of the involved matrices
only depend on the network topology and can be stored and reused for
each calculation run.

\section{The primal formulation of linearized network flows}

In this paper the linear `DC' approximation of the power flow in AC
networks is used, whose usefulness is discussed in \cite{Purc05,Hert06}.
In this section the linear power flow formulation for AC networks is
reviewed and a compact matrix notation is introduced.

In the linear approximation, the directed active power flow $F_\ell$ on a line $\ell$ from node $m$ to node $n$ can be expressed in terms of the line series reactance $x_\ell$ and the voltage angles $\theta_m, \theta_n$ at the nodes
\be
   F_{\ell} = \frac{1}{x_{\ell}} (\theta_m - \theta_n) = b_{\ell} (\theta_m - \theta_n) ,
   \label{eq:Ftheta}
\ee
where $b_{\ell} = 1/x_{\ell}$ is the susceptance of the line.
In the following we do not distinguish between transmission lines and transformers, which are treated the same.

The power flows of all lines are written in vector form, $\vec F = (F_1,\ldots,F_L)^t \in \mathbb{R}^L$, and similarly for the nodal voltage angles $\vec \theta = (\theta_1,\ldots,\theta_N)^t \in \mathbb{R}^N$, where the superscript $t$ denotes the transpose of a vector or matrix. Then equation \eqref{eq:Ftheta} can be written compactly in matrix form
\be
   \vec F = \vec B_d \vec K^t \vec \theta ,
   \label{eq:FthetaVector}
\ee
where $\vec B_d = \mbox{diag} (b_1,\ldots, b_L) \in \mathbb{R}^{L\times L}$
is the diagonal branch susceptance matrix.
The incidence matrix $\vec K \in \mathbb{R}^{N \times L}$ encodes how
the nodes of the directed network graph are connected by the lines \cite{Newm10}.
It has components
\be
   K_{n,\ell} = \left\{
   \begin{array}{r l}
      1 & \; \mbox{if line $\ell$ starts at node $n$},  \\
      - 1 & \; \mbox{if line $\ell$ ends at node $n$},  \\
      0     & \; \mbox{otherwise}.
  \end{array} \right.
  \label{eq:def-nodeedge}
\ee
In homology theory $\vec K$ is the boundary operator from the
vector space  of lines $\cong\mathbb{R}^{L}$ to the vector space of
nodes $\cong\mathbb{R}^{N}$.

The incidence matrix also relates the nodal power injections at each
node $\vec P = (P_1,\ldots,P_N) \in \mathbb{R}^N$ to the flows
incident at the node
\begin{equation}
  \vec P = \vec K \vec F . \label{KCL}
\end{equation}
This is Kirchhoff's Current Law expressed in terms of the active
power: the net power flowing out of each node must equal the power
injected at that node.

Combining \eqref{eq:FthetaVector} and \eqref{KCL}, we obtain
an equation for the power injections in terms of the voltage angles,
\begin{equation}
  \vec P = \vec B \vec \theta . \label{pinj}
\end{equation}
Here we have defined the nodal susceptance matrix
$
    \vec B \equiv \vec K \vec B_d \vec K^t  \, \in \mathbb{R}^{N \times N}
$
with the components
\begin{equation}
  B_{m,n} = \left\{
   \begin{array}{lll}
   \displaystyle\sum \nolimits_{\ell \in \Lambda_m} b_{\ell} &  \mbox{if } m = n; \\ [2mm]
     - b_{\ell} & \mbox{if }  m \mbox{ is connected to } n \mbox{ by } \ell,
   \end{array} \right. \label{eq:Bweighted}
\end{equation}
where $\Lambda_m$ is the set of lines which are incident on $m$. The matrix $\vec B$
is a weighted Laplace matrix \cite{Newm10} and equation \eqref{pinj} is a discrete
Poisson equation. Through equations \eqref{eq:FthetaVector} and \eqref{pinj}, there is
now a linear relation between the line flows $\vec F$ and the nodal power injections $\vec P$.

For a connected network, $\vec B$ has a single zero eigenvalue and
therefore cannot be inverted directly. Instead, the Moore-Penrose
pseudo-inverse $\vec B^*$ can be used to solve \eqref{pinj} for $\vec \theta$
and obtain the line flows directly as a linear function of the nodal power injections
\begin{equation}
    \vec F = \vec B_d \vec K^t \vec B^* \vec P. \label{linear}
\end{equation}
This matrix combination is taken as the definition of the nodal Power Transfer Distribution Factor (PTDF) $\mbox{\textbf{PTDF}} \in \mathbb{R}^{L \times N}$
\begin{equation}
  \mbox{\textbf{PTDF}} = \vec B_d \vec K^t \vec B^* .
  \label{eq:def-ptdf}
\end{equation}

Next, the effect of a line outage is considered. Suppose the flows
before the outage are given by $F_k$ and the line which
fails is labeled $\ell$. The line flows after the outage of $\ell$, $F^{(\ell)}_k$ are linearly related to the original flows by the matrix of Line Outage Distribution Factors (LODFs) \cite{Grai94,Wood14}
\be
   F^{(\ell)}_{k} = F_{k} +  \mbox{LODF}_{k\ell }  F_{\ell},
\ee
where on the right hand side there is no implied summation over $\ell$.
It can be shown \cite{Gule07,Guo09} that the LODF matrix elements can be expressed directly in terms of the PTDF matrix elements as
\be
    \mbox{LODF}_{k \ell} =\frac{[\mbox{\textbf{PTDF}}\cdot \vec K]_{k \ell}}{1 - [\mbox{\textbf{PTDF}}\cdot \vec K]_{\ell\ell}} \, .
    \label{eq:lodf-from-ptdf}
\ee
For the special case of $k = \ell$, one defines
$\mbox{LODF}_{kk} = -1$. The matrix $[\mbox{\textbf{PTDF}}\cdot \vec
  K]_{k \ell}$ can be interpreted as the sensitivity of the flow on
$k$ to the injection of one unit of power at the from-node of $\ell$
and the withdrawal of one unit of power at the to-node of $\ell$.

\section{Cycles and the Dual Graph}
\label{sec:graph}

The power grid defines a graph $G=(V,E)$ with vertex set $V$ formed by the nodes or buses and edge set $E$ formed by all transmission lines and transformers. The orientation of the edges is arbitrary but has to be fixed because calculations involve directed quantities such as the real power flow.
In the following we reformulate the theory of transmission line outages in terms of \emph{cycle flows}. A directed cycle is a combination of directed edges of the graph which form a closed loop. All such directed cycles can be decomposed into a set of $L-N+1$ fundamental cycles, with $N$ being the number of nodes, $L$ being the number of edges and assuming that the graph is connected \cite{Dies10}. An example is shown in Fig.~\ref{fig:5bus}, where two fundamental cycles are indicated by blue arrows.

The fundamental cyles are encoded in the cycle-edge
incidence matrix $\vec C \in \mathbb{R}^{L \times (L-N+1)}$
\be
C_{\ell, c} = \left\{
   \begin{array}{r l}
      1 & \; \mbox{if edge $\ell$ is element of cycle $c$},  \\
      - 1 & \; \mbox{if reversed edge $\ell$ is element of cycle $c$},  \\
      0     & \; \mbox{otherwise}.
  \end{array} \right.
  \label{eqn:cycle-edge-matrix}
\ee
It is a result of graph theory, which can also be checked by explicit calculation, that the $L-N+1$ cycles are a basis for the kernel of the incidence matrix $\vec K$ \cite{Dies10},
\be
\vec K \vec C = \vec 0 \, . \label{IC}
\ee

Using the formalism of cycles, the Kirchoff Voltage Law (KVL) can be expressed in a concise way.
KVL states that the sum of all angle differences along any closed cycle must equal zero,
\be
     \label{eq:phasecon}
     \sum_{(ij) \in {\rm cycle} \,  c}  \left( \theta_i - \theta_j\right) = 0 \, .
\ee
Since the cycles form a vector space it is sufficient to check this condition for the $L-N+1$ basis
cycles. In matrix form this reads
\be
     \label{eq:phaseconmatrix}
     \vec C^t \vec K^t \vec \theta  = 0 \, ,
\ee
which is satisfied automatically by virtue of equation \eqref{IC}.

Using equation \eqref{eq:FthetaVector}, the KVL
in terms of the flows reads
\be
\vec C^t \vec X_d \vec F = 0    \label{eq:uniqueness},
\ee
where $\vec X_d$ is the branch reactance matrix, defined by
$\vec X_d =\mbox{diag} (x_1,\ldots, x_L) = \mbox{diag} (1/b_1,\ldots, 1/b_L) \in \mathbb{R}^{L\times L}$.

The results of Sections \ref{sec:intro} through \ref{sec:comp} apply for any graph. In the final Section \ref{sec:top}, a special focus is made on planar graphs, i.e., graphs which can be drawn or `embedded' in the plane $\mathbb{R}^2$ without edge crossings. Once such an embedding is fixed, the graph is called a plane graph. Power grids are not naturally embedded in $\mathbb{R}^2$, but while line crossings are possible, they are sufficiently infrequent in large scale transmission grids (such as the high-voltage European transmission grid).

The embedding (drawing) in the plane yields a very intuitive approach to the cycle flow formulation. The edges separate polygons, which are called the \emph{facets} of the graph. We can now define a cycle basis which consists exactly of these facets. Then all edges are part of at most two basis cycles, which is called MacLane's criterion for panarity~\cite{Dies10}. This construction is formalized in the definition of the weak dual graph $DG$ of $G$. The weak dual graph $DG$ is formed by putting dual nodes in the middle of the facets of $G$ as described above, and then connecting the dual nodes with dual edges across those edges where facets of $G$ meet \cite{Dies10,Newm10}. $DG$ has $L-N+1$ nodes and its incidence matrix is given by $C^t$.

The simple topological properties of plane graphs are essential to derive some of the rigorous results obtained in
section~\ref{sec:topology}. For more complex networks, graph embeddings without line crossings can still be defined -- but not on the plane $\mathbb{R}^2$. More complex geometric objects (surfaces with a genus $g>0$)  are needed (see, e.g., \cite{Mode16} and references therein).

\section{Dual theory of network flows}

In this section the linear power flow is defined in terms of dual
cycle variables following \cite{ronellen16}, rather than the nodal voltage angles.
To do this, we define the linear power flow equations directly in terms of
the network flows. The power conservation equation \eqref{KCL}
\begin{equation}
  \vec K \vec F =   \vec P  , \label{eq:continuity}
\end{equation}
provides $N$
equations, of which one is linearly dependent, for the $L$ components
of $\vec F$. The solution space is thus given by an affine subspace of
dimension $L-N+1$.

In section \ref{sec:graph} we discussed that the kernel of $\vec
K$ is spanned by the cycle flows. Thus, we can write every solution of
equation (\ref{eq:continuity}) as a particular solution of the
inhomogeneous equation plus a linear combination of cycle flows:
\be
\vec F = \vec F^{(\rm part)} + \vec C \vec f, \qquad \vec f \in
\mathbb{R}^{L-N+1}.
\ee
The components $f_c$ of the vector $\vec f$
give the strength of the cycle flows for all basis cycles $c =
1,2,\cdots,L-N+1$. A particular solution $\vec F^{(\rm part)}$ can be found by taking the uniquely-determined flows on a spanning tree of the network graph \cite{ronellen16}.

To obtain the correct physical flows we need a further condition to
fix the $L-N+1$ degrees of freedom $f_c$. This condition is provided
by the KVL in \eqref{eq:uniqueness}, which provides exactly $L-N+1$ linear constraints on $\vec f$
\begin{equation}
   \vec C^t \vec X_d  \vec C \vec f = - \vec C^t \vec X_d \vec F^{(\rm part)}. \label{eq:dualPoisson}
\end{equation}
Together with equation (\ref{eq:continuity}), this condition uniquely
determines the power flows in the grid.

Equation \eqref{eq:dualPoisson} is the dual equation of \eqref{pinj}. If the cycle reactance matrix $\vec A\in \mathbb{R}^{L-N+1 \times L-N+1}$ is defined by
\begin{equation}
    \vec A \equiv \vec C^t \vec X_d \vec C,
    \label{eq:def-A}
\end{equation}
then $\vec A$ also has the form,
\begin{equation}
  A_{cc'} = \left\{
   \begin{array}{lll}
   \sum \limits_{\ell \in \kappa_c} x_{\ell} &  \mbox{if } c = c'; \\ [2mm]
     \sum \limits_{\ell \in \kappa_c \cap \kappa_{c'}} \pm x_{\ell} & \mbox{if }  c \neq c,
   \end{array} \right. \label{eq:Aexplicit}
\end{equation}
where $\kappa_c$ is the set of edges around cycle $c$ and the sign
ambiguity depends on the orientation of the cycles. The construction
of $\vec A$ is very similar to the weighted Laplacian in equation
\eqref{eq:Bweighted}; for plane graphs where the cycles correspond to
the faces of the graph, this analogy can be made exact (see Section
\ref{sec:topology}).  Unlike $\vec{B}$, the matrix $\vec{A}$ is
invertible, due to the fact that the outer boundary cycle of the
network is not included in the cycle basis.  This is analogous to
removing the row and column corresponding to a slack node from $\vec
B$, but it is a natural feature of the theory, and not manually imposed.

\section{Dual computation of line outage distribution factors}\label{sec:comp}

\subsection{Single line outages}

The dual theory of network flows derived in the previous section can be used to derive an alternative formula for the LODFs. For the sake of generality we consider an arbitrary change of the reactance of a transmission line $\ell$,
\be
   x_\ell \rightarrow x_\ell + \xi_\ell.
\ee
The generalization to multiple line outages is presented in the following section.
The change of the network structure is described in terms of the branch reactance matrix
\begin{align}
    \hat{ \vec X}_d &= \vec X_d + \Delta \vec X_d
     =   \vec X_d +  \xi_\ell \vec u_{\ell} \vec u_{\ell}^t,
\end{align}
where $\vec u_{\ell} \in \mathbb{R}^L$ is a unit vector which is 1 at position $\ell$ and zero otherwise.
In this section we use the hat to distinguish the line parameters and flows in the modified grid after the outage from the original grid before the outage.

This perturbation of the network topology will induce a change of the power flows
\be
   \hat{\vec F} = \vec F + \Delta \vec F.
\ee
We consider a change of the topology while the power
injections remain constant.
The flow change $\Delta \vec F$ thus does not have any source such that it
can be decomposed into cycle  flows
\be
    \Delta \vec F = \vec C \Delta \vec f.
    \label{eq:DF-cycledec}
\ee

The uniqueness condition (\ref{eq:uniqueness}) for the perturbed network reads
\be
  \vec C^t (\vec X_d + \Delta \vec X_d) (\vec F + \Delta \vec F) = \vec 0.
\ee
Using condition \eqref{eq:uniqueness} for the original network and the cycle flow
decomposition equation~\eqref{eq:DF-cycledec} for the flow changes yields
\begin{align}
         \label{eq:cyclecon1}
     & \vec C^t \hat{\vec X_d} \vec C \Delta\vec f = - \vec C^t  \, \Delta \vec X_d \vec F \nn \\
   \Rightarrow \quad & \Delta\vec f = - (\vec C^t \hat{\vec X_d} \vec C)^{-1} \vec C^t  \,
       \vec u_\ell \xi_\ell \vec u_\ell^t \vec F
\end{align}
such that the flow changes are given by
\be
    \Delta \vec F = \vec C \Delta\vec f = - \vec C (\vec C^t \hat{\vec X_d} \vec C)^{-1} \vec C^t  \,
       \vec u_\ell \xi_\ell \vec u_\ell^t \vec F.
     \label{eq:DF-from-G}
\ee
This expression suggests that we need to calculate the inverse separately for every possible contingency case, which would require a huge computational effort.
However, we can reduce it to the inverse of
$\vec A = \vec C^t \vec X_d \vec C$ describing the unperturbed grid
using the Woodbury matrix identity \cite{Wood50},
\begin{align*}
    & (\vec C^t \hat{\vec X_d} \vec C)^{-1} =
    \big( \vec A + \vec C^t \vec u_\ell \xi_\ell \vec u_\ell^t \vec C \big)^{-1} \\
    &= \vec A^{-1} - \vec A^{-1} \vec C^t \vec u_\ell \left( \xi_\ell^{-1} +
        \vec u_\ell^t \vec C \vec A^{-1} \vec C^t \vec u_\ell \right)^{-1} \vec u_\ell^t \vec C \vec A^{-1}.
\end{align*}
Thus we obtain
\begin{align*}
    &(\vec C^t \hat{\vec X_d} \vec C)^{-1} \vec C^t \vec u_\ell  = \vec A^{-1}\vec C^t \vec u_\ell
       \left( 1 + \xi_\ell \vec u_\ell^t \vec C \vec A^{-1} \vec C^t \vec u_\ell \right)^{-1}.
\end{align*}
We then obtain the induced cycle flows and flow change by inserting this expression into equation (\ref{eq:DF-from-G}). We summarize our results in the following proposition.

\begin{prop}
\label{prop:DeltaF}
If the reactance of a single transmission line $\ell$ is changed by an amount $\xi_\ell$, the real power flows change as
\be
     \Delta \vec F =
         \frac{- \xi_\ell F_\ell}{1 + \xi_\ell \vec u_\ell^t  \vec M \vec u_\ell }
         \,   \vec M  \vec u_\ell
\ee
with the matrix $\vec M = \vec C \vec A^{-1} \vec C^t$.
If the line $\ell$ fails, we have $\xi_\ell \rightarrow \infty$. The line outage distribution factor for a transmission line $k$ is thus given by
\be
\mbox{LODF}_{k,\ell} = \frac{\Delta F_k}{F_\ell} = - \frac{\vec u_k^t \vec M  \vec u_\ell}{
         \vec u_\ell^t  \vec M \vec u_\ell } \, .
   \label{eq:res-lodf-dual}
\ee
\end{prop}

Note that the formula for an arbitrary change in series reactance of a line is useful for the assessment of the impact of
flexible AC transmission (FACTS) devices, in particular series compensation devices \cite{Zhang2006} or adjustable inductors that clamp onto overhead lines \cite{SmartWires}.

\begin{figure}[tb]
\centering
\includegraphics[trim = 3.5cm 2.5cm 6.5cm 1.5cm,clip,width=8cm]{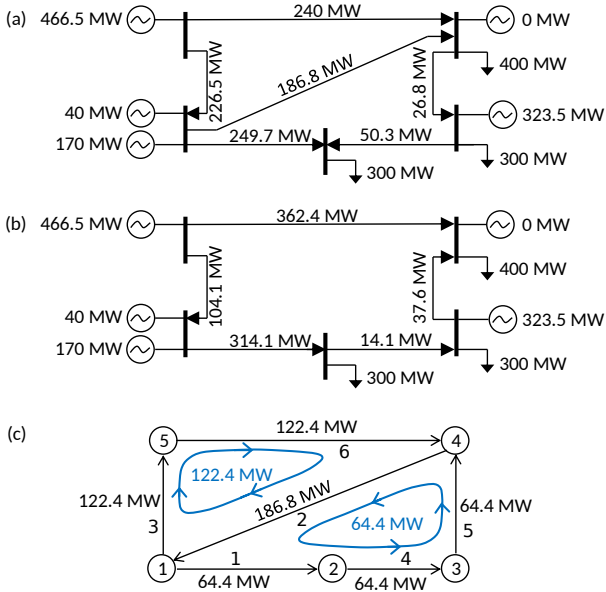}
\caption{
\label{fig:5bus}
(a) DC power flow in a 5-bus test grid from \cite{MATPOWER}.
(b) DC power flow in the same network after the outage of one transmission line.
(c) The change of power flows can be decomposed into two cycle flows.
}
\end{figure}

Finally, an example of failure induced cycle flows
and the corresponding flow
changes is shown in Figure~\ref{fig:5bus}.
In the example, the node-edge incidence matrix is given by
\be
    \vec K = \begin{pmatrix}
       +1 & -1 & +1 & 0 & 0 & 0 \\
       -1 & 0 & 0 & +1 & 0 & 0\\
       0 & 0 & 0 & -1 & +1 & 0 \\
       0 & +1 & 0 & 0 & -1 & -1 \\
       0 & 0 & -1 & 0 & 0 & +1 \\
       \end{pmatrix}.
\ee
The grid contains $2$ independent cycles, which are chosen as\\
\phantom{000} cycle 1: line 2, line 6, line 3. \\
\phantom{000} cycle 2: line 1, line 4, line 5, reverse line 2\\
The cycle-edge incidence matrix thus reads
\be
  \vec C^t = \begin{pmatrix}
        0 & +1 & +1 & 0 & 0 & +1\\
       +1 & -1 & 0 & +1 & +1 & 0 \\
       \end{pmatrix}.
\ee
Thus, the flow changes can be written according to
equation~\eqref{eq:DF-cycledec} with
\be
\Delta \vec f = (122.4~\mathrm{MW}, 64.4~\mathrm{MW})^t.
\ee
(cf.~also \cite{Cole16,Mani16} for a discussion of cycle flows in power grids).

The dual approach to the LODFs can be computationally advantageous for sparse networks as discussed in section~\ref{sec:compute}. Furthermore, we will use it to prove some rigorous results on flow redistribution after transmission line failures in section~\ref{sec:topology}.

\subsection{Multiple line outages}

The dual approach can be generalized to the case of multiple damaged or perturbed transmission lines in a straightforward way. Consider the simultaneous perturbation of the $M$ transmission lines
$\ell_1,\ell_2,\ldots,\ell_M$ according to
\be
   x_{\ell_1} \rightarrow x_{\ell_1}  + \xi_{\ell_1} ,
   x_{\ell_2} \rightarrow x_{\ell_2}  + \xi_{\ell_2} , \ldots,
   x_{\ell_M} \rightarrow x_{\ell_M}  + \xi_{\ell_M}. \nn
\ee
The change of the branch reactance matrix is then given by
\be
   \Delta \vec X_d = \vec \U \, \vec \Xi \, \vec \U^t,
   \label{eq:DeltaX}
\ee
where we have defined the matrices
\begin{align*}
   \vec \Xi &= {\rm diag}(\xi_{\ell_1}, \xi_{\ell_2}, \ldots, \xi_{\ell_M}) \; \in \mathbb{R}^{M \times M}, \\
    \vec \U &= (\vec u_{\ell_1}, \vec u_{\ell_2}, \ldots, \vec u_{\ell_M}  ) \; \in \mathbb{R}^{N \times M}.
\end{align*}
The formula (\ref{eq:DF-from-G}) for the flow changes then reads
\begin{align}
    \Delta \vec F &= - \vec C \left( \vec C^t \hat{\vec X}_d \vec C \right)^{-1}
          \vec C^t  \, \vec \U \, \vec \Xi \, \vec \U^t \vec F.
    \label{eq:DF-from-G-mult}
\end{align}
To evaluate this expression we again make use of the Woodbury matrix identity \cite{Wood50}, which yields
\begin{align*}
    & \left( \vec C^t \hat{\vec X}_d \vec C \right)^{-1}  = \nn \\
    & \quad  \vec A^{-1} - \vec A^{-1} \vec C^t \vec \U \left(\vec \Xi^{-1} +
        \vec \U^t \vec M \vec \U \right)^{-1} \vec \U^t \vec C \vec A^{-1}.
\end{align*}
We then obtain the flow change by inserting this expression into equation (\ref{eq:DF-from-G-mult}) with the result
\begin{align}
    \Delta \vec F &= - \vec C \vec A^{-1}  \vec C^t  \, \vec \U
        \left(\eye + \vec \Xi \, \vec \U^t \, \vec M \vec \U \right)^{-1}
       \vec \Xi \, \vec \U^t \vec F.
    \label{eq:DF-from-G-mult2}
\end{align}

In case of a multiple line outages of lines $\ell_1,\ldots,\ell_m$ we have to consider the limit
\be
   \xi_{\ell_1}, \ldots,\xi_{\ell_M} \rightarrow \infty.
\ee
In this limit equation (\ref{eq:DF-from-G-mult2}) reduces to
\begin{align}
    \Delta \vec F &= - \vec M  \, \vec \U
        \left(\vec \U^t \, \vec M \vec \U \right)^{-1}
       \vec \U^t \vec F.
    \label{eq:DF-from-G-mult-out}
\end{align}
Specifically, for the case of two failing lines, we obtain
\begin{align}
   \Delta F_k &= \frac{M_{k,\ell_1} M_{\ell_2,\ell_2} - M_{k,\ell_2} M_{\ell_2,\ell_1} }{M_{\ell_1,\ell_1}  M_{\ell_2,\ell_2} - M_{\ell_1,\ell_2} M_{\ell_2,\ell_1} } F_{\ell_1}  \nn \\
     & \quad + \frac{M_{k,\ell_2} M_{\ell_1,\ell_1} - M_{k,\ell_1} M_{\ell_1,\ell_2} }{M_{\ell_1,\ell_1}  M_{\ell_2,\ell_2} - M_{\ell_1,\ell_2} M_{\ell_2,\ell_1} } F_{\ell_2} \, .
\end{align}

\subsection{Computational aspects}
\label{sec:compute}

The dual formula (\ref{eq:res-lodf-dual}) for the LODFs can be computationally advantageous to the conventional approach. To calculate the LODFs via equation (\ref{eq:lodf-from-ptdf}) we have to invert the matrix $\vec B \in \mathbb{R}^{N \times N}$ to obtain the PTDFs. Using the dual approach the most demanding step is the inversion of the matrix $\vec A = \vec C^t \vec X_d \vec C \in \mathbb{R}^{(L-N+1) \times (L-N+1)}$, which can be much smaller than $\vec B$ if the network is sparse.
However, more matrix multiplications need to be carried out, which decreases the potential speed-up. We test the computational performance of the dual method by comparing it to the conventional approach, which is implemented in many popular software packages such as for instance in \textsc{Matpower} \cite{MATPOWER}.

Conventionally, one starts with the calculation of the nodal PTDF matrix defined in Eq.~(\ref{eq:def-ptdf}). In practice, one usually does not compute the full inverse but solves the linear system of equations ${\rm \bf PTDF} \cdot \, \vec B = \vec B_d \vec K$ instead. Furthermore, one fixes the voltage phase angle at a slack node $s$, such that one can omit the $s$th row and column in the matrix $\vec B$ and the $s$th column in matrix $\vec B_f = \vec B_d \vec K^T$ while solving the linear system.
The result is multiplied by the matrix $\vec K$ from the right to obtain the PTDFs between the endpoints of all lines. One then divides each column $\ell$ by the value $1 - \mbox{PTDF}_{\ell \ell}$ to obtain the LODFs via formula (\ref{eq:lodf-from-ptdf}).
An implementation of these steps in \textsc{Matlab} is listed in the supplement~\cite{supp}.

\begin{table}[t!]
\centering
\caption{
\label{tab:cputime}
Comparison of CPU time for the calculation of the PTDFs in sparse numerics.
}
\begin{tabular}{ |c|c|c|c|c|c| }
  \hline
  \multicolumn{2}{|c|}{Test Grid}
         & \multicolumn{3}{|c|}{Grid Size}
         & speedup \\
  name & source &
      nodes & cycles & ratio &
       \\
  & & $N$ & $\scriptstyle L-N+1$ & $\frac{L-N+1}{N}$
       &  $\frac{t_{\rm conv}}{t_{\rm dual}}$   \\
       \hline
 case300 & \cite{MATPOWER} & 300 & 110 & 0.37 &  $ 1.83 $ \\
 case1354pegase & \cite{Flis13} & 1354 & 357 & 0.26 & $ 4.43 $ \\
 GBnetwork & \cite{GBnet} & 2224 & 581 & 0.26 &  $ 4.09 $ \\
 case2383wp & \cite{MATPOWER} & 2383 & 504 & 0.21 & $ 4.20 $ \\
 case2736sp & \cite{MATPOWER} & 2736 & 760 & 0.28 & $ 3.27 $ \\
 case2746wp & \cite{MATPOWER} & 2746 & 760 & 0.28 & $ 3.35 $ \\
 case2869pegase & \cite{Flis13} & 2869 & 1100 & 0.38  & $ 2.79 $ \\
 case3012wp & \cite{MATPOWER} & 3012 & 555 & 0.18 & $ 3.93 $ \\
 case3120sp & \cite{MATPOWER} & 3120 & 565 & 0.18 &  $ 3.96 $ \\
 case9241pegase & \cite{Flis13} & 9241 & 4967 & 0.54  & $ 1.31 $ \\
      \hline
\end{tabular}
\end{table}

The dual approach yields the direct formula (\ref{eq:res-lodf-dual}) for the LODFs. To efficiently evaluate this formula we first compute the matrix $\vec M = \vec C \vec A^{-1} \vec C^t$. Again we do not compute the full matrix inverse but solve a linear system of equations instead. The full LODF matrix is then obtained by dividing every column $\ell$ by the factor $M_{\ell \ell}$.

We evaluate the runtime for various test grids from \cite{MATPOWER,Flis13,GBnet} using a \textsc{Matlab} program listed in the supplement~\cite{supp}. All input matrices are sparse, such that the computation is faster when using sparse numerical methods (using the command \texttt{sparse} in \textsc{Matlab} and converting back to \texttt{full} at the appropriate time). Then \textsc{Matlab} employs the high-performance supernodal sparse Cholesky decomposition solver \textsc{Cholmod} 1.7.0 to solve the linear system of equations. We observe a significant speed-up of the dual method by a factor between 1.31 and 4.43
depending on how meshed the grid is  (see Table~\ref{tab:cputime}).

\section{Topology of cycle flows}\label{sec:top}
\label{sec:topology}

In this section the propagation of the effects of line outages are analyzed using the theory of discrete calculus and differential operators on the dual network graph. There is a wide body of physics and mathematics literature on discrete field theory (see, e.g., \cite{Grady2010}). We turn back to the cycle flows themselves and derive some rigorous results. These results help to understand the effects of a transmission line outage and cascading failures in power grids, in particular whether the effects are predominatly local or affect also remote areas of the grid (cf.~\cite{Labavic2014,Kett15,Jung2015,Hine16}).

We start with a general discussion of the mathematical structure of the problem and show that line outages affect only parts of the grid which are sufficiently connected. Further results are obtained for planar graphs (graphs that can be embedded in the plane without line crossings, which approximately holds, e.g., for the European high voltage transmission grid). We characterize the direction of the induced cycle flows and show that the effect of failures decreases monotonically with the distance from the outage. Finally we proceed to discuss non-local effects in non-planar networks.

\subsection{General results}

The starting point of our analysis of the topology of cycle flows is a re-formulation of Proposition \ref{prop:DeltaF}.

\begin{lemma}
\label{lemma:cf1}
The outage of a single transmission line $\ell$ induces cycle flows which are determined by the linear system of equations
 \be
   \vec A \Delta\vec f = \vec q
   \label{eq:CyclePoisson}
\ee
with
$\vec q =  F_{\ell} (\vec u_{\ell}^t \vec C \vec A^{-1} \vec C^t \vec u_{\ell})^{-1} \vec C^t \vec u_{\ell}$
and  $\vec A = \vec C^t \vec X_d \vec C$.
\end{lemma}

Note that $\Delta\vec f, \vec q \in \mathbb{R}^{(L-N+1)}$ and $\vec A
\in \mathbb{R}^{(L-N+1) \times (L-N+1)}$. It will now be shown that in some cases this
equation can be interpreted as a discrete Poisson equation for $\Delta\vec f$ with
Laplacian operator $\vec A$ and inhomogeneity $\vec q$. This
formulation is convenient to derive some rigorous results on
flow rerouting after a transmission line failure.

We first note from the explicit construction of $\vec A$ in equation
\eqref{eq:Aexplicit} that two cycles in the dual network are only
coupled via their common edges. The coupling is given by the sum of
the reactances of the common edges. Generally, the reactance of a line is
proportional to its length. The coupling of two cycles is then
directly proportional to the total length of their common boundary,
provided that the lines are all of the same type.
Since the inhomogeneity $\vec q$ is proportional to $\vec C^t \vec u_{\ell}$, it is non-zero
only for the cycles which are adjacent to the failing edge $\ell$:
\be
  q_c \neq 0 \quad \mbox{only if $\ell$ is an element of cycle $c$}.
\ee

\begin{figure}[tb]
\includegraphics[width=0.48\columnwidth]{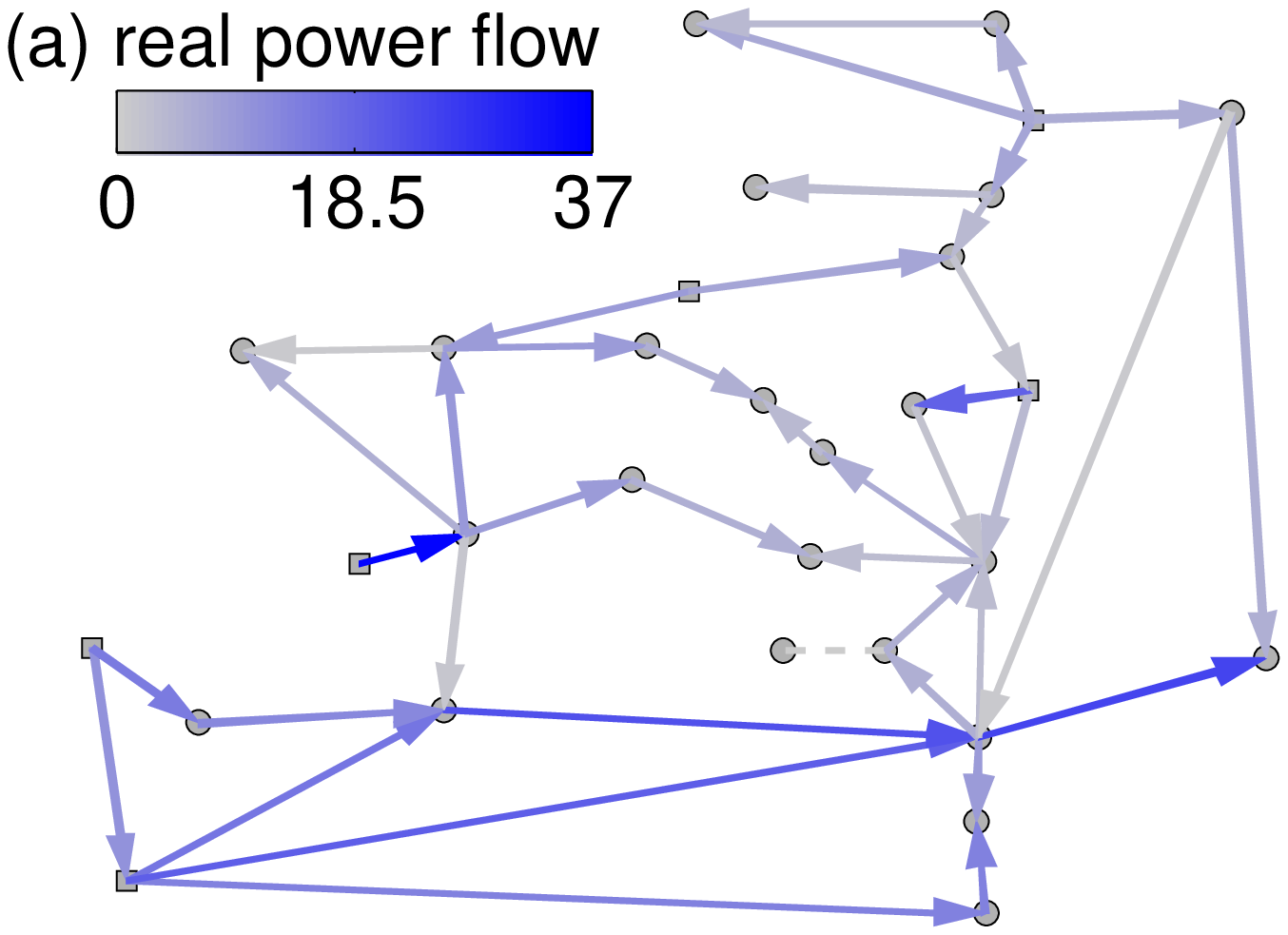}
\includegraphics[width=0.48\columnwidth]{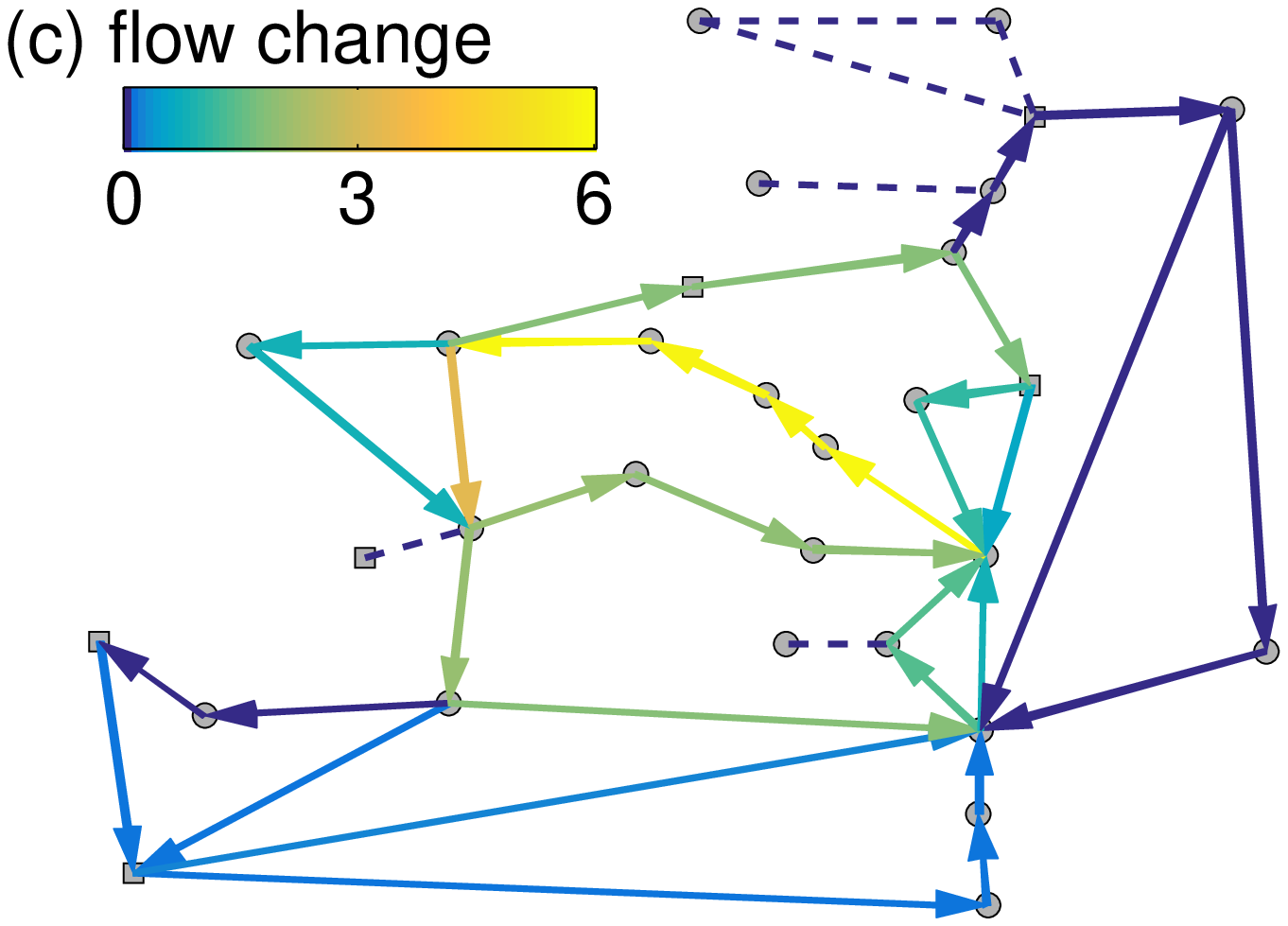} \\ [5mm]
\includegraphics[width=0.48\columnwidth]{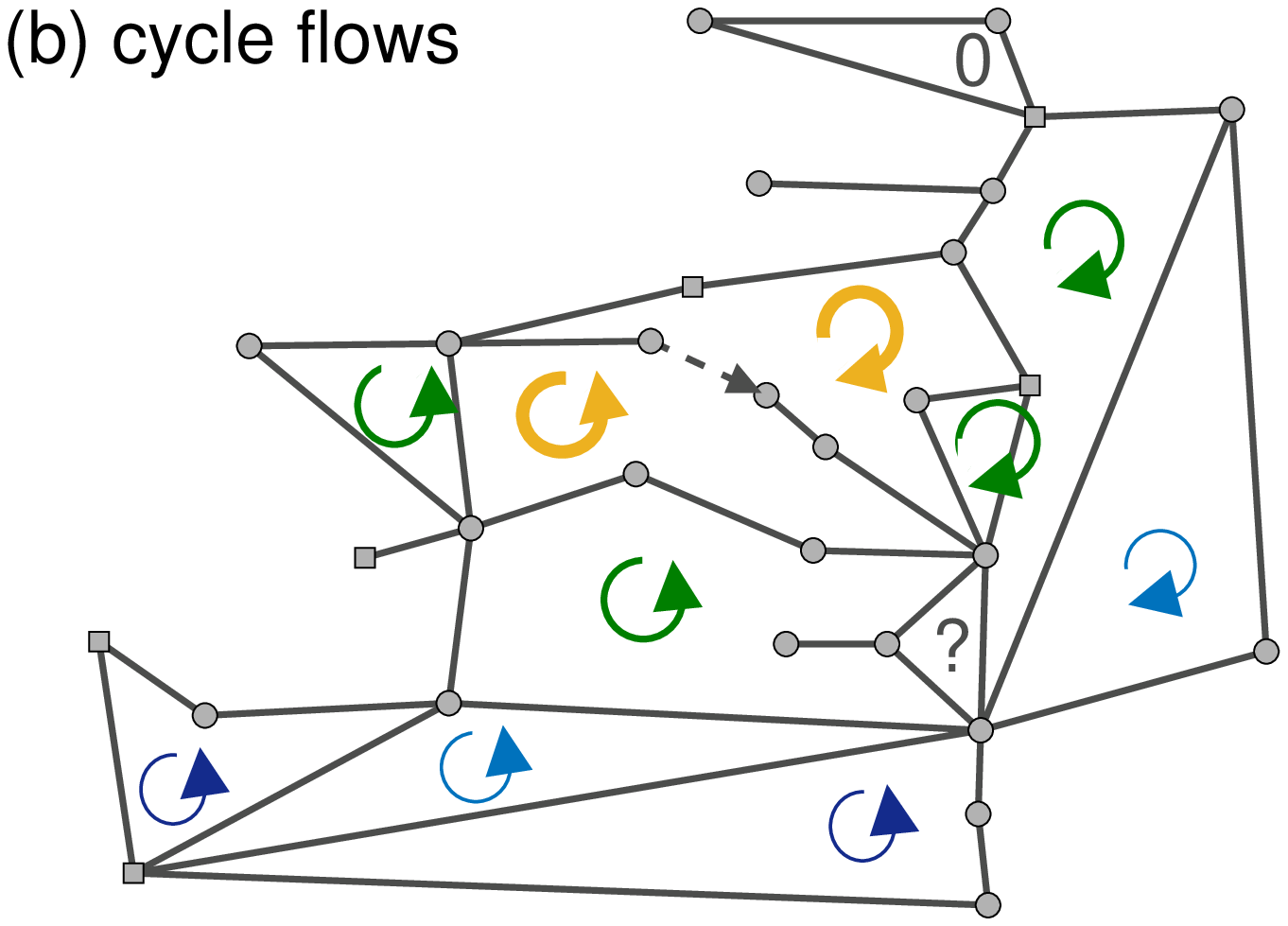}
\includegraphics[width=0.48\columnwidth]{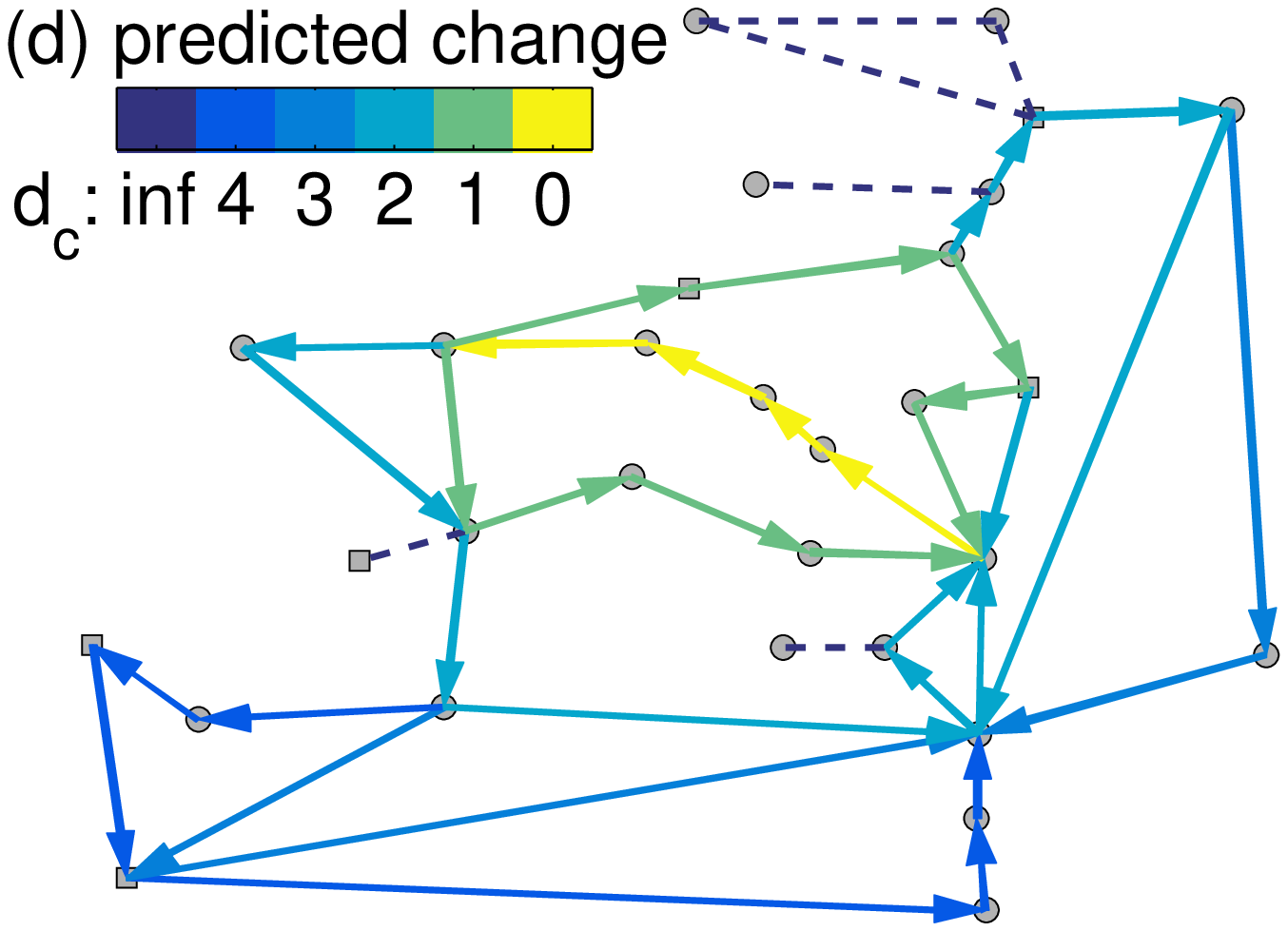}
\caption{
\label{fig:case30}
Flow changes and cycle flows after a transmission line failure in a small test grid.
(a) Real power flows in the initial intact network in MW.
(b) The failure of a transmission line (dashed) must be compensated by cycle flows indicated by
cyclic arrows. The thickness of the lines indicates the approximate strength of the cycle flows.
(c) The resulting flow changes after the failure of the marked transmission line.
(d) The direction of the flow changes can be predicted using Propositions \ref{prop:domains} and \ref{prop:decay}
for all edges and the magnitude decreases with the cycle distance $d_c$.
The power flow in (a,c) has been calculated using the standard software
\textsc{Matpower} for the 30-bus test case \cite{MATPOWER}.
}
\end{figure}

The matrix $\vec A$ typically has a block structure such that a failure in one block cannot affect the remaining blocks. The dual approach to flow rerouting gives a very intuitive picture of this decoupling. To see this, consider the example shown in Figure \ref{fig:case30}. The cycle at the top of the network is connected to the rest of the network via one node. However, it is decoupled in the dual representation because it shares no common edge with any other cycle. Thus, a failure in the rest of the grid will not affect the power flows in this cycle---the mutual LODFs vanish. This result is summarized in the following proposition, and a formal proof is given in the
supplement~\cite{supp}.

\begin{prop}
\label{prop:decoupling}
The line outage distribution factor $\mbox{LODF}_{k,\ell}$ between two edges
$k = (i,j)$ and $\ell = (s,r)$ vanishes if there is only one independent path
between the vertex sets $\{r,s\}$ and $\{i,j\}$.
\end{prop}

\subsection{Planar networks}

Some important simplifications can be made in the specific case of a plane network.
We can then define the cycle basis in terms of the interior \emph{faces}
of the graph which allows for a intuitive geometric picture of induced cycle flows
as in Figures \ref{fig:case30} and \ref{fig:scheme}. For the remainder of this section
we thus restrict ourselves to such plane graphs and fix the cycle basis by the
interior faces and fix the orientation of all basis cycles to be counter-clockwise.
Thus equation (\ref{eq:CyclePoisson}) is formulated on the \emph{weak dual}
of the original graph.

\begin{figure}[tb]
\centering
\includegraphics[trim=2cm 2cm 6cm 1cm, clip, width=8cm]{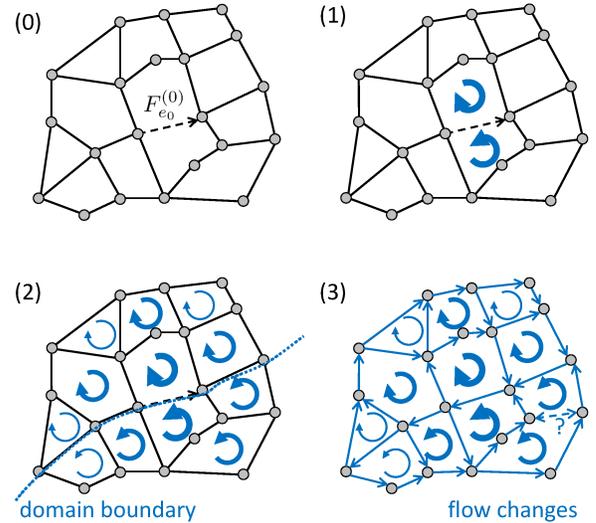}
\caption{
\label{fig:scheme}
Schematic representation of the
flow changes after the damage of a single edge (dashed).
}
\end{figure}

According to Mac Lane's planarity criterion~\cite{Dies10}, every edge in a plane graph belongs
to at most two cycles such that $\vec q$ has at most two-nonzero elements:
One non-zero element $q_{c_1}$ if $\ell$ is at the boundary and two non-zero
elements $q_{c_1} = - q_{c_2}$ if the line $\ell$ is in the interior of the network.
Furthermore, the matrix $\vec A$ is a Laplacian matrix in the interior of the
network~\cite{Newm10}. That is, for all cycles $c$ which are not at the boundary we have
\be
   \sum_{d \neq c} A_{dc} = - A_{cc}.
   \label{eq:Laplace-A}
\ee
Up to boundary effects, equation (\ref{eq:CyclePoisson}) is thus equivalent to a
discretized Poisson equation on a complex graph with a dipole source
(monopole source if the perturbation occurs on the boundary).

For plane networks we now prove some rigorous results on the orientation
of cycle flows (clockwise vs. counter-clockwise) and on their decay
with the distance from the failing edge.
In graph theory, the (geodesic) distance of two vertices is defined as as
the number of edges in a shortest path connecting them \cite{Dies10}.
Similarly, the distance of two edges is defined as the number of
vertices on a shortest path between the edges.

\begin{prop}
\label{prop:domains}
Consider the cycle flows $\Delta\vec f$ induced by the failure of a a single line $\ell$
in a plane linear flow network described by equation (\ref{eq:CyclePoisson}).
The weak dual graph can be decomposed into at most two connected subgraphs
(`domains') $\DD_+$ and $\DD_-$, with $\Delta f_c \ge 0 \,  \forall c \in  \DD_+$ and
$\Delta f_c \le 0 \,  \forall c \in  \DD_-$. The domain boundary, if it exists, includes the perturbed
line $\ell$, i.e. the two cycles adjacent to $\ell$ belong to different domains.
\end{prop}

A proof is given in the supplement~\cite{supp}. The crucial aspect of this proposition
is that the two domains $\DD_+$ and $\DD_-$ must be \emph{connected}. The
implications of this statement are illustrated in Figure \ref{fig:scheme} in panel (2), showing
the induced cycle flows when the dashed edge is damaged. The induced cycle flows
are oriented clockwise above the domain boundary and counter-clockwise below
the domain boundary. If the perturbed edge lies on the boundary of a finite plane network,
then there is only one domain and all cycle flows are oriented in the same way.

With this result we can obtain a purely geometric view of how the flow of all edges in the
network change after the outage. For this, we need some additional information about
the magnitude of the cycle flows in addition to the orientation. We consider the upper
and lower bound for the cycle flows $\Delta f_c$ at a given distance to the cycle
$c_1$ with $q_{c_1} > 0$ and the cycles  $c_2$ with $q_{c_2} < 0$, respectively:
\begin{align}
   u_d   &= \max_{c, {\rm dist}(c,c_1) = d}  \Delta f_c \nn \\
  \ell_d &=  \min_{c, {\rm dist}(c,c_2) = d}  \Delta f_c.
\end{align}
Here, $\text{dist}$ denotes the graph-theoretic distance between two cycles or faces,
i.e. the length of the shortest path between the two faces in the dual graph.
We then find the following result.


\begin{prop}
\label{prop:decay}
The maximum (minimum) value of the cycle flows decreases (increases) monotonically with the
distance $d$ to the reference cycles $c_1$ and $c_2$, respectively:
\begin{align}
   & u_d  \le u_{d-1}, \qquad 1 \le d \le d_{\rm max}.  \nn \\
   & \ell_d  \ge \ell_{d-1},  \qquad 1 \le d \le d_{\rm max}.
\end{align}
\end{prop}
A proof is given in the supplement~\cite{supp}.
Strict monotonicity can be proven when some additional technical assumptions are satisfied,
which are expected to hold in most cases.
For a two-dimensional lattices with regular topology and constant weights the cycle flows
are proportional to the inverse distance (see supplement~\cite{supp} for details). However, irregularity of the network
topology and line parameters can lead to a stronger, even exponential, localization \cite{Kett15,Jung2015,Labavic2014}.
Hence, the response of the grid is strong only in the `vicinity' of the damaged
transmission line, but may be non-zero everywhere in the connected component.

However, it has to be noted that the distance is defined
for the dual graph, not the original graph, and that the rigorous results hold only for
plane graphs. The situation is much more involved in non-planar graphs, as a line
can link regions which would be far apart otherwise. Examples for the failure
induced cycles flows and the decay with the distance are shown in Figure
\ref{fig:case30}.


\subsection{General, non-planar networks}

Here, we consider fully general, non-planar networks. Unlike in the previous section, we show that it is impossible to derive a simple monotonic decay of the effect of line failures. Instead, by decomposing the LODFs into
a geometric and a topological part, we show that complex, non-local interactions result. We start with
\begin{prop}
\label{prop:topo}
Every connected graph $G$ can be embedded into a Riemannian surface of genus $g \in \mathbb{N}_0$ without line crossings. The cycle basis can be chosen such that it consists of the boundaries of $L-N+1-2g$ geometric facets of the embedding encoded in the cycle adjacency matrix $\vec{\tilde C} \in \mathbb{R}^{L \times (L-N+1-2g)}$ and $2g$ topological non-contractible cycles encoded in the cycle adjacency matrix $\vec{\hat C} \in \mathbb{R}^{L \times 2g}$, which satisfy $\vec{\hat C}^t \vec X_d \vec{\tilde C} = \vec 0$
and $\vec{\tilde C}^t \vec X_d \vec{\hat C} = \vec 0$.
\end{prop}

A proof is given in the supplement~\cite{supp}. The main result of this proposition is that the cycle basis of any graph can be decomposed into two parts. The geometric cycles behave just as the facets in a planar graph. But for non-planar graphs there is a second type of cycles -- the topological ones. For the simplest non-planar examples one can find an embedding without line-crossings on the surface of a torus, which has the genus $g=1$. Two topological cycles have to be added to the cycle basis, which wind around the torus in the two distinct directions. These cycles are intrinsically non-local. The following corollary now shows that also the effects of a line outage can be decomposed.

\begin{corr}
Consider a general graph with embedding and cycle basis as in Proposition \ref{prop:topo}. Then the flow changes after the outage of a line $\ell$ are given by
\be
   \Delta \vec F = \vec{\tilde C} \Delta \vec{\tilde f} + \vec{\hat C} \Delta \vec{\hat f}
\ee
where the cycle flows are given by
\begin{align}
    (\vec{\tilde C}^t \vec X_d \vec{\tilde C})  \Delta \vec{\tilde f} &= \frac{F_{\ell}}{M_{\ell, \ell}} \vec{\tilde C}^t \vec u_{\ell} \label{eq:decomp-geom} \\
       (\vec{\hat C}^t \vec X_d \vec{\hat C})  \Delta \vec{\hat f} &= \frac{F_{\ell}}{M_{\ell, \ell}} \vec{\hat C}^t \vec u_{\ell}. \label{eq:decomp-topol}
\end{align}
\end{corr}

\begin{proof}
According to proposition \ref{prop:topo} the cycle incidence matrix is decomposed as $\vec C = \begin{pmatrix} \vec{\tilde C},\vec{\hat C} \end{pmatrix}$. Similarly, we can decompose the strength of the cycle flows after the line outage as
\be
   \Delta \vec f = \begin{pmatrix} \Delta \vec{\tilde f} \\ \Delta \vec{\hat f} \end{pmatrix}
\ee
such that the flow changes are given by $\Delta \vec F = \vec{\tilde C} \Delta \vec{\tilde f} + \vec{\hat C} \Delta \vec{\hat f}$. Then Eq.~\eqref{eq:CyclePoisson} reads
\be
   \begin{pmatrix} \vec{\tilde C}^t \\ \vec{\hat C}^t \end{pmatrix} \vec X_d \,
    \begin{pmatrix} \vec{\tilde C},\vec{\hat C} \end{pmatrix}
    \begin{pmatrix} \Delta \vec{\tilde f} \\ \Delta \vec{\hat f} \end{pmatrix}
   = \frac{F_\ell}{M_{\ell,\ell}} \begin{pmatrix} \vec{\tilde C}^t \\ \vec{\hat C}^t \end{pmatrix} \vec u_\ell.
\ee
Using that $\vec{\hat C}^t \vec X_d \vec{\tilde C} = \vec 0$
and $\vec{\tilde C}^t \vec X_d \vec{\hat C} = \vec 0$ the corollary follows.
\end{proof}

Remarkably, the corollary shows that the cycle flows around geometric and topological cycles can be decoupled. The matrix $\vec{\tilde A} =  \vec{\tilde C}^t \vec X_d \vec{\tilde C}$ has a Laplacian structure as in Eq.~\eqref{eq:Laplace-A} because at each edge of the graph at most two facets meet. Thus, Eq.~\eqref{eq:decomp-geom} is a discrete Poisson equation as for plane graphs and the propositions \ref{prop:domains} and \ref{prop:decay} also hold for for the flows $\Delta \vec{\tilde f}$ around the geometric cycles.
However, Eq.~\eqref{eq:decomp-topol} has no such interpretation and it is, in general, dense on both sides. Thus, the topological cycles represented by Eq.~\eqref{eq:decomp-topol} are responsible for complicated, non-local effects of damage spreading in general power grids.

\section{Conclusions}

Line Outage Distribution Factors are important for assessing the reliability
of a power system, in particular with the recent rise of renewables.
In this paper, we described a new dual formalism for calculating LODFs,
that is based on using power flows through the closed
cycles of the network instead of using nodal voltage angles.

The dual theory yields a compact formula for the LODFs that only
depends on real power flows in the network.
In particular, the formula lends itself to a straightforward
generalization for the case of multiple line outages.
Effectively, using cycle flows instead of voltage angles
changes the dimensionality of the matrices appearing in
the formulae from $N\times N$ to $(L-N+1)\times (L-N+1)$.
In cases where the network is very sparse (i.e., it contains
few cycles but many nodes), this can lead to a significant speedup
in LODF computation time, a critical improvement for quick
assessment of real network contingencies. In addition, the formalism generalises easily to multiple outages and arbitrary changes in series reactance, which is important for the assessment of the impact of FACTS devices.
Often, some of the quantities involved in power flow problems
are not known exactly, i.e., they are random (see, e.g.,~\cite{Muehlpfordt2016}). Thus, extending
our work to include effects of randomness will
be an important next step.

The dual theory not only yields improvements for numerical computations,
it also provides a novel viewpoint of the underlying physics
of power grids, in particular if they are (almost) planar.
Within the dual framework for planar networks, it is easy to show that single line
contingencies induce flow changes in the power grid which decay monotonically
in the same way as an electrostatic dipole field.

\appendices

\section*{Acknowledgments}
We gratefully acknowledge support from
the Helmholtz Association (joint initiative `Energy System 2050 -- a contribution of the research field energy' and
grant no.~VH-NG-1025 to D.W.) and
the German Federal Ministry of Education and
Research (BMBF grant nos.~03SF0472B, ~03SF0472C and ~03SF0472E).
The work of H. R. was supported in part by the IMPRS Physics of Biological
and Complex Systems, G\"ottingen.

\end{document}